\documentclass[onecolumn]{IEEEtran}
\usepackage{amsmath,amssymb,amsthm,mathtools}
\usepackage{cite}
\usepackage{verbatim}
\usepackage[hidelinks]{hyperref}
\newcommand{\R}{\mathbb{R}}
\newcommand{\GOE}{\operatorname{GOE}}
\newcommand{\avg}[1]{\left\langle #1\right\rangle}
\newcommand{\dd}{\mathop{}\!\mathrm{d}}

\newtheorem{theorem}{Theorem}
\newtheorem{lemma}[theorem]{Lemma}
\newtheorem{proposition}[theorem]{Proposition}
\newtheorem{corollary}[theorem]{Corollary}
\newtheorem{remark}[theorem]{Remark}

\title{On the Cost of Entrainment in Protein Translation}

\author{Ram~Massas and Michael~Margaliot%
\thanks{The authors are with the School of Electrical and Computer  Engineering,
Tel Aviv University, Tel Aviv 69978, Israel.
Corresponding author: Michael Margaliot (michaelm@tauex.tau.ac.il).}%
\thanks{This research was partially supported by the Israel Science
Foundation under Grant 221/24.}}

\begin{document}

\maketitle
\begin{abstract}
    Biological systems often synchronize their dynamics with periodic environmental and intracellular signals. Whether such periodic  coordination can also improve performance, however, remains unclear. Here, we study this question in the ribosome flow model, a nonlinear dynamical model of ribosome movement along an mRNA transcript during translation. We compare the average protein production rate under positive periodic transition rates with that of a constant-rate system obtained by replacing each rate by its temporal average. We prove that periodic modulation can never increase average protein production: the gain of entrainment is always nonpositive. Moreover, the gain is zero if and only if all transition rates share a common positive periodic modulation of their mean values. This exceptional modulation merely reparametrizes time and leaves the trajectory through state space unchanged. Thus, any genuinely nonuniform temporal modulation strictly reduces average protein production. Our results establish a fundamental tradeoff between temporal coordination and translational efficiency: entrainment can synchronize translation with periodic cellular programs, but this synchronization comes at a quantifiable production cost.
\end{abstract}


\section{Introduction}
Biological systems are 
exposed to changing 
environments. Nutrient availability, cellular resources, and regulatory signals often vary in a periodic pattern. Such temporal variation can synchronize molecular processes with recurring environmental or intracellular programs e.g. the 24h solar day. But can periodic  coordination also improve performance?

To address this question rigorously, 
consider a nonlinear dynamical system with state vector \(x\), input \(u\), and a scalar output
\[
y=h(x,u).
\]
We say that the system \emph{entrains} if, for every \(T\)-periodic input~\(u\), every solution~\(x(t)\) converges to a unique \(T\)-periodic orbit~\(\gamma \). Consequently, the output converges to the \(T\)-periodic signal
\[
y^\gamma(t):=h(\gamma(t),u(t)).
\]

Entrainment is important in both natural and artificial systems, as it allows synchronization with periodic external or internal signals. Synchronous generators must entrain to the frequency of the grid. Biological processes like  the wake-sleep cycle  must entrain to the 24h solar day. Suppose now that  our goal is to maximize the average output
\[
\frac{1}{T}\int_0^T h(\gamma(t),u(t))dt. 
\]
For example, the output may represent the 
flow of vehicles  along a road
controlled by  $T$-periodic traffic lights, and our goal is to improve the average flow. 

A natural question is whether 
entrainment can also improve the output, on average.  The notion of gain of entrainment~(GOE) defines this rigorously  as follows.
Let~$\avg{u}:=\frac1{T}\int_0^T u(t)dt$, that is, the average input. Since the system entrains, for this input the state vector converges to an equilibrium~$e$, and the output converges to
the constant value 
\[
 h(e,\avg{u}).
\]
The gain of entrainment for the control~$u$  is defined  as
\[
\GOE(u):=
\frac{1}{T}\int_0^T h(\gamma(t),u(t))dt-h(e,\avg{u}).
\]
Thus, a positive GOE implies that entrainment not only synchronizes the system to periodic excitations, but also improves the output, on average. Since constant controls are a special case of $T$-periodic control, one may suspect that GOE is always nonnegative, but this is not true. 

Analyzing  the
GOE is an  important theoretical question with many potential  applications including   periodic fishery,
photo-bioreactors     modulated using periodic light,   traffic flow regulated  using periodic traffic lights, medications administered in a periodic pattern, and more. 

The difficulty  in analyzing GOE stems from the fact that in nonlinear systems the 
equilibrium~$e$ and the periodic  orbit~$\gamma$ are not known  explicitly.  GOE can also be posed as an  optimal control theory, 
but solutions exist only in simple cases, e.g. scalar systems~\cite{rapo_convex}.
Ref.~\cite{katz2023gainentrainmentclassweakly} considered  a specific class of weakly contractive systems, and 
 derived a closed-form formula for~GOE to first-order in the control perturbation. This was  used to show  that~GOE is always a second- or higher-order phenomenon.  
Ref.~\cite{massas2026gainentrainmentstablelinear}  analyzed GOE in the special  class of systems with a linear dynamics and a nonlinear output function, and related GOE to the convexity/concavity of the output function.

Here, we analyze GOE in the 
  ribosome flow model~(RFM) which is a deterministic nonlinear model for the 
unidirectional flow of particles along  a chain of~$n$ ordered sites. The~RFM is the dynamic  
mean-field approximation of
a fundamental model from statistical mechanics called the totally asymmetric simple exclusion process~(TASEP)~\cite{solvers_guide,kriecherbauer_krug2010}. The RFM 
has been extensively 
used to model and analyze ribosome movement along   the mRNA molecule (see e.g.~\cite{reuveni2011genome,margaliot2012stability,kaminer2026turnpikepropertyeigenvalueoptimization,rfm_max,rfm_sense,EYAL_RFMD1}). More recently, networks of interconnected~RFMs have been used to study large-scale translation in the cell and the effect of competition for shared resources like ribosomes and initiation
factors~\cite{Raveh2016,fierce_compete,aditi_networks}.

Gene expression  is 
 excited by periodic signals  due to circadian rhythms 
 and the periodic cell-cycle program~\cite{goldbeter_book}. 
Furthermore,
synthetic genetic oscillators have  become an important paradigm for studying  and designing  biological clocks~\cite{elowitz2000synthetic,purcell2010comparative}.
Entrainment to periodic signals is important in gene expression because it enables biological systems to synchronize their internal dynamics with recurring  periodic  signals. 
This motivates the study of the~RFM with positive transition rates that are
time-varying and jointly $T$-periodic,  
e.g. due to the periodic 
cell-cycle program or 
engineered periodic control.
It is known that the periodic~RFM entrains: there exists a  unique $T$-periodic solution of the system that is globally asymptotically stable~\cite{RFM_entrain} (for more   on entrainment in biological systems, see~\cite{entrain_trans}).
In other words,  in response to $T$-periodic 
transition rates, due e.g. to the periodic cell-cycle program or periodic demands on shared resources, the ribosome densities and the protein production rate in the RFM will converge to a unique~$T$-periodic pattern. 

Previous work~\cite{GOE_RFM_2023} has shown that in several  special cases~GOE
in the~RFM is nonnegative,
leading to the conjecture  that~GOE 
in the~RFM is always nonnegative.
The main difficulty in proving this conjecture stems from the
fact that 
 the~RFM is a nonlinear dynamical model, and the time average of its periodic solution is not generally the
equilibrium associated with the averaged rates.  

This paper resolves  this  conjecture. Its contributions are threefold: 
(1)~we derive an exact identity representing the GOE in the~RFM 
as a weighted sum of
edge-wise dissipation terms;
(2)~we prove that every dissipation  term is
nonnegative   yielding~$\GOE\leq0$ for
arbitrary positive and jointly $T$-periodic rates; and (3)~we explicitly  characterize the case where the GOE is zero, namely, 
 when all the rates are obtained from their averages
by a single common positive periodic scalar modulation. Such a modulation only
changes the speed at which the constant-rate dynamics is traversed. Thus, in any non degenerate  case, GOE in the~RFM is negative implying that periodic modulation   always decreases  the average production rate. 
This establishes a fundamental tradeoff between temporal coordination and production efficiency: entrainment can synchronize translation with periodic cellular signals, but   this always  comes with a cost in terms of protein production. 

The remainder of the paper is organized as follows. Section~\ref{sec:pri} reviews  the
periodic RFM~(PRFM) and the problem of~GOE in the~PRFM.  Section~\ref{sec:main} states the main result. Section~\ref{sec:proof} derives
the dissipation identity,  proves nonpositivity of GOE and characterizes the case of zero~GOE.    To demonstrate the theoretical results,
a numerical example is provided in Section~\ref{sec:numerical}. 
The final section concludes and describes possible  directions for further research.

\section{Model and Problem Formulation}\label{sec:pri}
In this section, we   review~GOE in the~PRFM.

\subsection{ Periodic Ribosome Flow Model}
The $n$-dimensional RFM
describes the unidirectional flow of particles along an ordered chain of~$n$
sites. The density of particles at site~$i$ at time~$t$ is described by the state
variable~$x_i(t)$, which is normalized such that~$x_i(t)\in[0,1]$ for all time~$t$, with~$x_i(t)=0$ 
[$x_i(t)=1$] corresponding to the case where the site is completely free [full].

The   rate of flow from site~$i-1$ to site~$i$ at time $t$ is given by
\[
\lambda_{i-1}(t)x_{i-1}(t)\bigl(1-x_i(t)\bigr),
\]
where~$\lambda_{i-1}(t)$
is the transition rate 
between the sites at time~$t$.
The flow  rate is thus proportional to the density at site~$i-1$ and the   ``free space'' $1-x_i(t)$ in site~$i$. 
In particular, if site~$i$ begins to fill up i.e.~$x_i(t) $
goes to one,  then the flow rate into site~$i$ goes to zero, 
and particles will start to accumulate in site~$i-1$, and then at site~$i-2$ and so on. Thus, the RFM provides a natural framework for modeling and analyzing the evolution of ``traffic jams'' of particles.

In the context of translation, the particles are ribosomes and the chain of sites is the mRNA molecule divided into (groups of) consecutive codons.  Ribosomal traffic jams, which arise when slow or stalled ribosomes impede trailing ribosomes   can lead to collisions, impaired translational efficiency, and activation of ribosome-associated quality-control pathways. Such processes have been implicated in a range of pathological conditions, including neurodegenerative diseases and cancer~\cite{schuller2018roadblocks,FILBECK20221451,subramaniam2021ribosome}.

The RFM is described by~$n$ nonlinear ordinary differential equations:
\begin{align}
\dot{x}_i(t)
&=\lambda_{i-1}(t)x_{i-1}(t)\bigl(1-x_i(t)\bigr) -\lambda_i(t)x_i(t)\bigl(1-x_{i+1}(t)\bigr),
\quad i=1,\ldots,n,
\label{eq:rfm}
\end{align}
with the boundary conventions
\begin{equation*}
x_0(t):=1,
\qquad
x_{n+1}(t):=0.
\end{equation*}
Eq.~\eqref{eq:rfm} states that the change in density in site~$i$ is the flow from site~$i-1$ to site~$i$ minus the flow from site~$i$ to site~$i+1$. 

The flow exiting the last site, namely,
\[
R(t):=\lambda_n(t)x_n(t)
\]
is the protein production rate at time $t$, since each ribosome exiting the last site completes translation and releases one protein molecule.

Note that~$x_i$ has no units, and the~$\lambda_i$s have units of~$1/\text{time}$. 
The state-space of the RFM is the unit cube~$[0,1]^n$ and any solution emanating from~$x(0)\in [0,1]^n$ satisfies~$x(t)\in [0,1]^n$ for all~$t\geq 0$.

 When using the RFM to model ribosome flow along a transcript, the transition rates~$u_i$ are determined by various biophysical features of the translation process, e.g. the abundance of cognate tRNA molecules, and   the availability and activity of translation factors (see e.g.~\cite{reuveni2011genome,Zur2020}). These rates may therefore vary along the transcript and may also change in response to the physiological state of the cell. Consequently, translation can be viewed as a dynamical process in which the effective transition rates are subject to temporal modulation. An important question is  whether periodic changes in the cellular environment can enhance or impair protein synthesis.

Fix~$n+1$ 
continuous, positive  and $T$-periodic functions of time~$u_0(t),\dots,u_n(t)$. 
In the periodic RFM~(PRFM), the transition rates are
  $\lambda_i(t)=u_i(t)$ for all~$i$. 
In particular, all the transition rates are jointly~$T$-periodic. 
The PRFM admits a unique~$T$-periodic solution
\begin{equation*}
\gamma(t)
:=\begin{bmatrix}
\gamma_1(t)&\dots&\gamma_n(t)
\end{bmatrix}^{\top},
\qquad
\gamma(t+T)=\gamma(t),
\end{equation*}
with $\gamma_i(t)\in(0,1)$ for all~$t$,  and   for any initial
condition~$x(0)\in[0,1]^n$ the solution of the~PRFM converges to~$\gamma$
(see~\cite{RFM_entrain}).

For a $T$-periodic 
function~$f$, let
\[
\avg{f}:=\frac1{T}\int_0^T f(t)\,dt
\]
be the time average  of~$f$. Then the average production rate in the~PRFM is 
\begin{equation*}
R_P:=\avg{u_n\gamma_n},
\end{equation*}
where the subscript
``P'' stands for periodic. 

\subsection{GOE in the RFM}
We compare the average  protein production rate~$R_P$ in the~PRFM to the production rate obtained in the constant-rates RFM,  where each transition rate is  the constant function~$\lambda_i(t) =\avg{u_i}$ for all~$t\geq 0$. In other words, every periodic transition rate is replaced by a constant function which is its time average. 

   The constant-rates RFM  admits a unique equilibrium
\begin{equation*}
e:=\begin{bmatrix}e_1&\dots&e_n\end{bmatrix}^{\top},
\qquad e_i\in(0,1),
\end{equation*}
and the solution emanating from any~$x(0)\in[0,1]^n$ converges to~$e$~\cite{margaliot2012stability}.
At equilibrium, 
the flows into and out of each site are equal, so~\eqref{eq:rfm} gives 
\begin{align}\label{eq:stedy}
\avg{u_0} (1-e_{1})& =\avg{u_1} e_1(1-e_{2})\nonumber\\
&\vdots\\
&=\avg{u_{n-1}} e_{n-1}(1-e_{n})\nonumber\\
&=
\avg{u_n} e_n.\nonumber
\end{align}
The 
 steady-state production rate is 
$
R_C:=\avg{u_n} e_n, 
$
where the subscript ``C'' stands for
 constant.  For each edge, introduce
the steady state current
\begin{align*}
c_i&:=e_i(1-e_{i+1}) ,\quad i=0,\dots,n.
\end{align*}
Then~$c_i\in(0,1)$, and
\begin{equation}
R_C=\avg{u_i} c_i,
\qquad i=0,\ldots,n.
\label{eq:Rc-aici}
\end{equation}

The gain of entrainment for the $T$-periodic control~$u=\begin{bmatrix}
    u_0&\dots&u_n\end{bmatrix}$ is
\begin{equation*}
\GOE(u):=R_P-R_C.
\end{equation*}
Thus, a positive GOE would mean that periodic modulation  yields a larger
average output than constant modulation  with the same average rates. In other words, entrainment not only helps  synchronizing  translation with periodic excitations, but also increases the  average translation efficiency. 

The RFM can also be interpreted as a model of one-dimensional vehicular traffic, with the transition rates representing the effects of traffic lights that vary periodically in time, with a common  period~$T$. In this setting, it is natural to expect that a clever coordination of the temporal traffic patterns could increase the average flow, corresponding to a positive ~GOE. Surprisingly, our main result shows that this intuition is false: periodic modulation of the transition rates can never increase the average flow, as the   GOE is always nonpositive.

\section{Main Result}\label{sec:main}

\begin{theorem}[Periodic modulation cannot increase average production]
\label{thm:main}
For any dimension~$n$ and any set of continuous, positive and~$T$-periodic translation rates~$u_i(t)$, $i=0,\dots,n$, we have   
\begin{equation*}
\GOE(u)\leq 0,
\end{equation*}
with equality  only when all the periodic rates can be written as
\[
u_i(t)=m(t) r_i,\quad i=0,\dots, n,
\]
where every~$r_i>0$ is a constant, and~$m(t)$ is a   positive $T$-periodic function.
\end{theorem}

Equivalently~$R_P\leq R_C$, with equality only in the degenerate  case   
where   all the rates  share one common $T$-periodic  modulation. Thus, any genuinely nonuniform temporal modulation produces a strict loss in average production.

The proof of this result, given in the next section, is not trivial because
    the periodic orbit generated by the periodic time-varying rates is generally not known explicitly. Thus the proof is not based on some form of pointwise comparison, but rather on a  global dissipation identity coupling the entire chain.

\section{Proof of the Main Result}\label{sec:proof}

The proof includes  four steps. First, periodicity implies that the time-averaged flows through all edges are equal. Second, we compare the periodic orbit with the equilibrium of the averaged system and introduce weights that produce telescoping cancellations along the chain. Third, these cancellations yield an exact dissipation identity expressing the GOE as a weighted sum of local dissipation terms. Finally, each local dissipation term is shown to be nonnegative yielding~$\GOE(u)\leq 0$.
 
\subsection{Average Flow Balance}
Define  
\begin{equation*}
e_0:=1,\qquad \gamma_0(t):=1,
\qquad
e_{n+1}:=0,\qquad
\gamma_{n+1}(t):=0.
\end{equation*}

For each  $i=0,\ldots,n$, define the  flow from site~$i$ to site~$i+1$ along the unique~$T$-periodic solution~$\gamma$ by 
 \begin{equation*}
q_i(t):=u_i(t)\gamma_i(t)\bigl(1-\gamma_{i+1}(t)\bigr). 
\end{equation*}
In particular, $q_0(t)=u_0(t)(1-\gamma_1(t))$,
and $q_n(t)=u_n(t)\gamma_n(t)$.

\begin{lemma}[Equal average edge flows]
\label{lem:equal-average-flows}
The average flows $\avg{q_i}$ across all edges are equal, and satisfy 
\begin{equation*}
\avg{q_0}=\avg{q_1}=\cdots=\avg{q_n}=R_P.
\end{equation*}
\end{lemma}

\begin{IEEEproof}
Along the periodic solution,
\begin{equation}
\dot{\gamma}_i(t)=q_{i-1}(t)-q_i(t),
\qquad i=1,\ldots,n.
\label{eq:periodic-conservation}
\end{equation}
Averaging over one period and using the fact that~$\gamma_i(T)=\gamma_i(0)$ yields
$\avg{q_{i-1}}=\avg{q_i}$. Combining this with the fact that~$R_P= \avg{q_n}=\avg{u_n\gamma_n}$
completes the proof. 
\end{IEEEproof}

\subsection{Centered Variables and Edge Expansions}
It is convenient to work with certain ``centered''  variables. We denote these variables by an overbar.
Define the \emph{centered    periodic solution}  by
\begin{align*}
\bar \gamma_i(t):=\gamma_i(t)-e_i,
\quad i=0,\ldots,n+1. 
\end{align*}
Note that
\[
\bar\gamma_0(t)=\bar\gamma_{n+1}(t)=0
.
\]
Define 
the \emph{centered transition rates}
by
\begin{align*}
\bar u_i(t)&:=u_i(t)-\avg{u_i},
\quad i=0,\ldots,n.
\end{align*}
Then
$
\avg{\bar u_i}=0$.
Define also the \emph{centered currents}
along the periodic orbit 
by
\begin{equation}
\bar q_i(t):=q_i(t)-R_C,\quad i=0,\dots,n.
\label{eq:delta-def}
\end{equation}
By Lemma~\ref{lem:equal-average-flows}, 
\begin{equation}
\avg{\bar q_i}=R_P-R_C ,
\qquad i=0,\ldots,n , 
\label{eq:delta-average}
\end{equation}
which is~$\text{GOE}(u)$. To analyze the sign of  this expression, first note that
\begin{align*}
q_i&=u_i \gamma_i(1-\gamma_{i+1})\\
&=u_i(e_i+\bar \gamma_i)(1-e_{i+1}-\bar\gamma_{i+1}) \nonumber\\
&=u_i(c_i+h_i-\bar \gamma_i \bar\gamma_{i+1}).
\end{align*}
where we define
\begin{equation}
    \label{eq:defofhi}
h_i(t) :=(1-e_{i+1})\bar\gamma_i(t)-e_i \bar\gamma_{i+1}(t).
\end{equation}
Note that~$h_0(t)=-\bar\gamma_1(t)$ and~$h_n(t) =\bar\gamma_n (t)$.
Using $u_i=\avg{u_i}+\bar u_i$,  and~\eqref{eq:Rc-aici} gives
\begin{equation}
\bar q_i=c_i\bar u_i+u_i h_i-u_i \bar\gamma_i \bar\gamma_{i+1}.
\label{eq:delta-identity}
\end{equation}
 Combining this with~\eqref{eq:delta-average} implies that
 \begin{equation*}
 \GOE(u)=\avg{\bar q_i}=\avg{u_i h_i}-\avg{u_i\bar\gamma_i\bar\gamma_{i+1}},
 \end{equation*}
 where we used the fact that~$\avg{\bar u_i}=0$. However, determining the sign of this expression seems difficult. It turns out that it is better to consider a sum in the form
 \begin{align}
     \label{eq:sumpi}
 \sum_{i=0}^n p_i\GOE(u)&=\sum_{i=0}^n p_i \avg{\bar q_i}\nonumber\\
 &= \sum_{i=0}^n p_i  \avg{u_ih_i}-\sum_{i=0}^n p_i\avg{ u_i\bar\gamma_i\bar\gamma_{i+1} } ,
 \end{align}
 where the~$p_i$s are certain positive weights that  sum to one, and  then simplify the sums in~\eqref{eq:sumpi}.

\subsection{Positive Weights and  Three  Cancellation Identities}
For an initial value $p_0>0$, to be specified below, define  iteratively 
\begin{equation}
p_i:=p_{i-1}\frac{e_i}{1-e_i},
\qquad i=1,\ldots,n.
\label{eq:p-recursion}
\end{equation}
Then all the~$p_i$s  are  positive. Now choose~$p_0$ such  that
\begin{equation}
\sum_{i=0}^{n}p_i=1.
\label{eq:p-normalization}
\end{equation}

Using the~$p_i$s
transforms sums of properly weighted~$h_i$s into a telescopic sum. To show this, 
first note that  
\begin{align*}
c_i^{-1}h_i&=\frac{(1-e_{i+1})\bar\gamma_i-e_i \bar\gamma_{i+1}}{e_i(1-e_{i+1})}\\& = e_i^{-1}\bar\gamma_i-(1-e_{i+1})^{-1} \bar\gamma_{i+1}
,
\end{align*}
so
  for~$i=0,\dots,n-1$, we have 
\begin{align}\label{eq:symm}
p_i c_i^{-1} h_i&= p_i e_i^{-1} \bar\gamma_i-p_i (1-e_{i+1})^{-1} \bar\gamma_{i+1}
\nonumber\\
&=p_i e_i^{-1} \bar\gamma_i-p_{i+1} e_{i+1}^{-1} \bar\gamma_{i+1}.
\end{align}
 
\begin{lemma}[First Weighted Cancellation]
\label{lem:weighted-cancellation_of_ci}
For the  weights defined by~\eqref{eq:p-recursion} and~\eqref{eq:p-normalization}, we have 
\begin{equation}\label{eq:anosim}
\sum_{i=0}^{n}p_i c_i^{-1} h_i(t) 
=0, \quad\text{for every }t\geq 0. 
\end{equation}
\end{lemma}

\begin{IEEEproof}
 This follows from~\eqref{eq:symm} and the boundary conditions~$\bar\gamma_0(t)=0$ and~$h_n(t)=\bar\gamma_n(t)$.    
\end{IEEEproof}

To  derive another  useful  telescopic  sum, note that  for any~$i=1,\dots,n$, we have 
\begin{align}\label{eq:poty}
p_i \avg{u_i}h_i&= p_i\avg{u_i}     (1-e_{i+1}) {e_i}{e_i}^{-1}\bar\gamma_i- p_i\avg{u_i}  e_i \bar\gamma_{i+1} \nonumber\\
&=p_i \avg{u_{i-1}} 
(1-e_{i})e_{i-1}e_i^{-1}\bar\gamma_i-p_i \avg{u_i}  e_i \bar\gamma_{i+1}\nonumber\\
&=p_{i-1}\avg{u_{i-1}} e_{i-1}
 \bar\gamma_i - p_i\avg{u_i}  e_i \bar\gamma_{i+1} , 
\end{align}
where the second equality follows from~\eqref{eq:stedy}, 
and the third from~\eqref{eq:p-recursion}.

\begin{lemma}[Second Weighted Cancellation]
\label{lem:weighted-cancellation_of_ui}
For the  weights defined by~\eqref{eq:p-recursion} and~\eqref{eq:p-normalization}, we have 
\begin{equation*}
\sum_{i=0}^{n}p_i \avg{u_i} h_i(t)=0,
\quad\text{for every }t\geq 0.
\end{equation*}
\end{lemma}
\begin{IEEEproof}
This follows from~\eqref{eq:poty} 
and the fact that~$\bar\gamma_{n+1}(t)=0$.
\end{IEEEproof}

\begin{lemma}[Third Weighted Cancellation]
\label{lem:third_weighted-cancellation}
For the  weights defined by~\eqref{eq:p-recursion} and~\eqref{eq:p-normalization}, we have 
\begin{equation}\label{eq:delta-h-balance}
\sum_{i=0}^{n}p_i c_i^{-1}
\avg{ \bar q_i h_i} =0.
\end{equation}
\end{lemma}

\begin{IEEEproof}
  Define
\begin{equation*}
\theta_i(t):={p_i}{e_i}^{-1}\bar\gamma_i(t),
\quad i=1,\ldots,n,
\end{equation*}
with boundary conditions
\begin{align}\label{eq:bnd_theta}
\theta_0(t)=\theta_{n+1}(t):=0.
\end{align}
 Then~\eqref{eq:symm} and the boundary conditions give
 \begin{equation}
\theta_i-\theta_{i+1}= {p_i}{c_i}^{-1}h_i,
\qquad i=0,\ldots,n.
\label{eq:theta-edge}
\end{equation}

Define the nonnegative   function
\begin{equation*}
V(t):=\frac{1}{2}\sum_{i=1}^{n} p_i {e_i}^{-1}\bar\gamma_i^2(t).
\end{equation*}
Since~$\bar\gamma_i(t)=\gamma_i(t)-e_i$,
 this is   a weighted ``quadratic energy function''  measuring how far the periodic solution is from the constant-rate equilibrium.
 Now,
 \begin{align*}
     \dot V &= \sum_{i=1}^n {p_i}{e_i}^{-1}\bar\gamma_i  \dot{\bar \gamma}_i\\
     &= \sum_{i=1}^n \theta_i  \dot{\bar \gamma}_i\\
     &=\sum_{i=1}^n \theta_i  \dot{ \gamma}_i,
 \end{align*}
 and using~\eqref{eq:periodic-conservation} gives~$\dot V =\sum_{i=1}^n \theta_i  (q_{i-1}-q_i)$. Combining this with~\eqref{eq:bnd_theta} implies that
 \[
 \dot V =\sum_{i=0}^n q_i  (\theta_{i+1}-\theta_i),
 \]
and~\eqref{eq:theta-edge} gives
 \begin{align*}
\dot V
& = -\sum_{i=0}^{n}{p_i}c_i^{-1}q_i h_i.
\end{align*}
Using~\eqref{eq:delta-def} yields 
$\dot V
  = -\sum_{i=0}^{n}p_ic_i^{-1}(  \bar q_i +R_C ) h_i$,  and~\eqref{eq:anosim}
gives
\[
\dot V
  = -\sum_{i=0}^{n}{p_i}c_i^{-1}   \bar q_i  h_i.
  \]
By the definition of~$V$, we have~$V(T)=V(0)$, that is,~$\avg{\dot V}=0$, and this completes the proof.  
\end{IEEEproof}

Note that combining~\eqref{eq:sumpi} and Lemma~\ref{lem:weighted-cancellation_of_ui} 
gives
 \begin{align}     \label{eq:yter}
\GOE(u)
 &= \sum_{i=0}^n p_i  \avg{\bar u_ih_i}-\sum_{i=0}^n p_i\avg{ u_i\bar\gamma_i\bar\gamma_{i+1} } ,
 \end{align}
The next step is to rewrite  the first  sum   on the right-hand side of this equation.

\begin{lemma}
    We have
\begin{align}\label{eq:wh-identity}
 \sum_{i=0}^n p_i \avg{\bar u_i h_i} & =\sum_{i=0}^n p_i  c_i^{-1} \avg{ u_i  h_i  \bar\gamma_i \bar\gamma_{i+1}} -
\sum_{i=0}^n p_i  c_i^{-1}\avg{ u_i  h_i^2 }
 .
\end{align}
\end{lemma}
\begin{IEEEproof}
Multiplying~\eqref{eq:delta-identity} by $p_i c_i^{-1} h_i $
gives
\begin{equation*}
p_i c_i^{-1} h_i \bar q_i=p_i \bar u_i   h_i + p_i  c_i^{-1} u_i     h_i^2  
-p_i  c_i^{-1}  u_i    h_i   \bar\gamma_i \bar\gamma_{i+1},
\end{equation*}
so
\[
p_i c_i^{-1} \avg{ h_i \bar q_i}=p_i\avg{ \bar u_i   h_i} + p_i  c_i^{-1} \avg{u_i     h_i^2}  
-p_i  c_i^{-1}  \avg{u_i    h_i   \bar\gamma_i \bar\gamma_{i+1}}.
\]
Summing  from~$i=0$ to~$i=n$, and using~\eqref{eq:delta-h-balance}  completes the proof.
\end{IEEEproof}

 We can now give a more useful expression for the~GOE. 
\begin{proposition}[Dissipation representation of the GOE]
\label{prop:dissipation}
For $i=0,\ldots,n$, define
\begin{equation}d_i(t):=c_i^{-1}{h_i^2(t)}
+ 
\left(1-{c_i}^{-1}{h_i(t)} \right)\bar\gamma_i(t)\bar\gamma_{i+1}(t).
\label{eq:D-def}
\end{equation}
Then
\begin{equation}
\GOE(u)=-\sum_{i=0}^{n}p_i\avg{u_i d_i}.
\label{eq:exact-dissipation}
\end{equation}
\end{proposition}

 \begin{IEEEproof}
This follows from combining~\eqref{eq:yter}
and~\eqref{eq:wh-identity}.
 \end{IEEEproof}

Thus, GOE is exactly the negative of the total weighted ``dissipation''
terms along the edges of the chain. Note that since every~$h_i$ 
and every~$\bar\gamma_i$  is~$T$-periodic, so is every~$d_i$.

\subsection{Nonnegativity of the Edge Dissipation}

\begin{lemma}
\label{lem:D-nonnegative}
For every $i=0,\ldots,n$, we have 
\begin{equation}
d_i(t)\geq0   \text{ for all } t\geq 0.
\label{eq:D-nonnegative}
\end{equation}
\end{lemma}

\begin{IEEEproof}
By definition, 
\begin{align}
\label{eq:dois}
d_0(t) = \frac{h_0^2(t)}{c_0}
+
\left(1-\frac{h_0(t)}{c_0}\right) \bar\gamma_0(t)\bar\gamma_{1}(t)= \frac {\bar \gamma_1^2(t)}{1-e_1} \geq 0
\end{align}
and
\[
d_n(t) =\frac{h_n^2(t)}{c_n}
+
\left(1-\frac{h_n(t)}{c_n}\right) \bar\gamma_n(t)\bar\gamma_{n+1}(t)=\frac{\bar\gamma_n^2(t)}{e_n}\geq 0.   
\]
Fix an internal edge $i\in\{1,\ldots,n-1\}$. 
Define 
\begin{equation*}
\begin{aligned}
\alpha&:=e_i,
&\qquad
\beta&:=1-e_{i+1},\\
r(t)&:=\frac{\bar\gamma_i(t)} {\alpha},
&
s(t)&:=\frac{\bar\gamma_{i+1}(t)}{\beta}.
\end{aligned}
\end{equation*}
Since $\gamma_i(t)=\alpha(1+r(t))>0$ and
$1-\gamma_{i+1}(t)=\beta(1-s(t))>0$, we have
\begin{equation*}
r(t)>-1,
\quad
s(t)<1 \text{ for all } t\geq 0.
\end{equation*}
Moreover,
\begin{equation*}
c_i=\alpha\beta,
\qquad
h_i=c_i(r-s),
\qquad
\bar\gamma_i \bar\gamma_{i+1}=c_i rs.
 \end{equation*}
so~\eqref{eq:D-def} gives 
\begin{align*}
d_i
&=c_i(r-s)^2+c_i rs(1-r+s) \nonumber\\
&=c_i F(r,s),
\end{align*}
where
\begin{equation*}
F(r,s):=r^2(1-s)+s^2(1+r)-rs.
\end{equation*}
It remains to prove that \begin{align}
    \label{eq:f_lower}
    F(r,s)\geq0 \text{ on the domain } \{(r,s):r>-1,\, s<1\}.
\end{align}
    To prove this, consider the problem of minimizing~$F(r,s)$ over the  domain~$\{(r,s):r\geq -1,\, s\leq 1\}$. Note that~$F(0,0)=0$. 
Define:
\[
x:=r+1,\, y:=1 - s.
\]
Then the problem is to minimize
\[
F(x,y)= 1-xy(3-x-y)
\]
over the domain~$\{(x,y):x\geq 0,\, y\geq 0\}$. 
Thus this transformation converts the constrained domain into the nonnegative orthant, 
 and reduces the problem to maximizing a symmetric cubic expression.

It is clear that the minimization 
problem is well-defined. We consider two cases.

\noindent Case 1: Suppose that~$x+y\geq 3$. Then~$F(x,y)\geq 1$. 

\noindent Case 2: Suppose that~$x+y<3$. Let~$z:=x+y$. Then~$z \in [0,3)$. Since~$xy\leq z^2/4$,  we have 
\[
xy(3-x-y)\leq \frac{z^2(3-z)}{4}.
\]
On the domain~$z \in [0,3)$, we have that~$g(z):=\frac{z^2(3-z)}{4}$ attains a unique maximum at~$z=2$ with~$g(2)=1$, so
\[
xy(3-x-y)\leq 1,
\]
and thus 
\[
F(x,y)  \geq 0. 
\]
We conclude that~\eqref{eq:f_lower} indeed holds, and this proves~\eqref{eq:D-nonnegative}. 
\end{IEEEproof}

    We can now prove the first assertion in Theorem~\ref{thm:main}. 
By Proposition~\ref{prop:dissipation}, 
$\GOE(u)=-\sum_{i=0}^{n}p_i\avg{u_i d_i}$.
Since~$p_i >0$,  $u_i(t)>0$ for all~$t\geq 0$,  and~$d_i(t)\geq0$ for all~$t\geq 0$,    every term in
the sum is nonnegative, and hence
\[
\GOE(u)\leq0.
\]

\subsection{Characterization of the Zero GOE Case}\label{sec:zerogoe}
The next result shows that~$\GOE(u)=0$ only in a specific and degenerate  case where all the $T$-periodic transition rates  are equal up to scaling by a positive constant. 
\begin{proposition}
[Characterization of zero GOE]
\label{thm:zero-goe}
Under the assumptions of Theorem~\ref{thm:main}, the following statements are
equivalent:
\begin{enumerate}
\item $\GOE(u)=0$.
\item There exists a continuous and
positive $T$-periodic function
$m:\R_{\geq0}\to\R_{>0}$ such that
\begin{equation}
u_i(t)=m(t)\avg{u_i},
\, i=0,\ldots,n,
\text{ with} \avg{m}=1.
\label{eq:common-modulation}
\end{equation}
\end{enumerate}
\end{proposition}
\begin{IEEEproof}
Assume first that $\GOE=0$. By~\eqref{eq:exact-dissipation},
$
\sum_{i=0}^{n}p_i\avg{u_id_i}=0.$ The term~$u_i(t)d_i(t)$ is continuous and nonnegative for all~$t\geq 0$, and~$p_i>0$ for all~$i$, so we conclude that 
\begin{equation}
d_i(t)\equiv0,
\qquad i=0,\ldots,n.
\label{eq:all-D-zero}
\end{equation}
In particular,~\eqref{eq:dois} gives $\bar\gamma _1(t) = 0$ for all~$t\geq 0$. 
We now show that this implies that every~$\bar \gamma_k$ is identically zero. 
Suppose
that for some 
index~$i\in\{1,\ldots,n-1\}$ we 
have
\[
\bar\gamma_i(t)=0 \text{  for all }t\geq 0.
\]
Then~\eqref 
{eq:defofhi}
gives
\begin{equation*}
h_i(t) =-e_i \bar\gamma_{i+1}(t),
\end{equation*}
so
\begin{equation*}
d_i(t)=\frac{e_i^2}{c_i}\bar\gamma_{i+1}^2(t),
\end{equation*}
and~\eqref{eq:all-D-zero} implies that~$\bar\gamma _{i+1}(t)\equiv0$. 
We conclude that 
\begin{equation*}
\bar\gamma_i(t) =0,
\text{ for all } i=1,\ldots,n, \text{ and all }t\geq 0,
\end{equation*}
so the unique periodic solution is constant:~$\gamma(t)\equiv e$.
Substituting this  in~\eqref{eq:rfm} gives
\begin{equation*}
u_{i-1}(t)c_{i-1}=u_i(t)c_i,
\qquad i=1,\ldots,n.
\end{equation*}
Thus, all the transition rates  are equal up to scaling by a positive constant.  Defining~$m(t):=u_0(t)/ \avg{u_0}$ yields~\eqref{eq:common-modulation}.

To prove the converse implication, assume 
that~\eqref{eq:common-modulation}  holds. Then 
the~RFM equations become 
\begin{align*}
\dot{x}_i(t)
=m(t)\bigl(
\avg{u_{i-1}}x_{i-1}(t)(1-x_i(t) ) -\avg{u_i} x_i(t)(1-x_{i+1}(t))\bigr).
\end{align*}
Since~$m$ is a continuous and positive function, this is just a time-parametrization of the RFM with constant transition rates~$\lambda_i=\avg{u_i}$. Thus, the common modulation changes only the speed at which trajectories are
traversed; it does not change the state-space trajectories nor the equilibrium.
In particular, the unique~$T$-periodic  solution is just~$e$, and thus~$R_P=R_C$.
\end{IEEEproof}

This completes the proof of Theorem~\ref{thm:main}. 

\section{Numerical Example}
\label{sec:numerical}

Consider an RFM with $n=3$ sites
and  rates~$\lambda_i(t)=u_i(t)$ where 
\begin{equation*}
u_0(t)=1,
\qquad
u_1(t)=2+\cos(t),
\qquad
u_2(t)=2,
\qquad
u_3(t)=3
 \end{equation*}
(see Figure~\ref{fig:peri_controls}).
Note that these rates are jointly $T$-periodic with~$T=2\pi$,
and that they do not satisfy~\eqref{eq:common-modulation}, so the theory developed above implies~$\GOE(u)<0$. 

\begin{figure}[!t]
\centering
\includegraphics[scale=0.6]{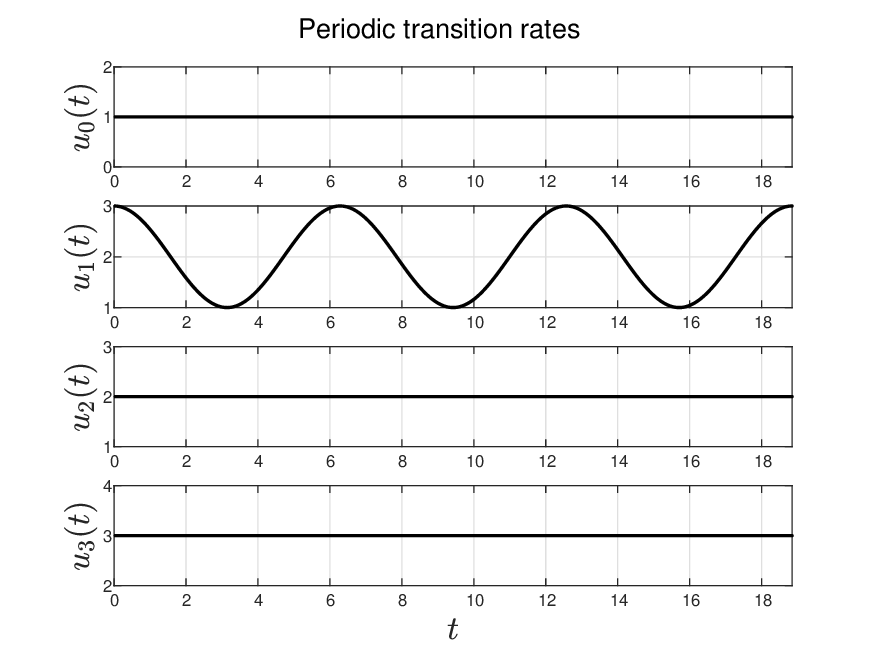}
\caption{ Periodic transition rates~$u_0(t),\dots,u_3(t) $
as a function of~$t\in[0,3T)$. }
\label{fig:peri_controls}
\end{figure}

The corresponding $T$-periodic solution~$\gamma$ was computed
numerically,\footnote{All the numerical calculations  were performed using Matlab.}
and numerically  integrating along this solution gives
\[
R_P=\frac1{2\pi}\int_0^{2\pi} u_3(t) \gamma_3(t)dt =0.5521 
\]
(all numerical values 
are to four-digit accuracy). 

The  equilibrium~$e$ corresponding  to the averaged rates
\begin{equation*}
\lambda_0=\avg{u_0}=1,
\qquad
\lambda_1=\avg{u_1}=2,
\qquad
\lambda_2=\avg{u_2}=2,
\qquad
\lambda_3=\avg{u_3}=3
\end{equation*}
was computed  using~\eqref{eq:stedy} yielding
\[
e_1=    0.4343
,\qquad
e_2=0.3486 ,\qquad
e_3=    0.1886 ,
\]
  so
\[
R_C=\avg{u_3}e_3= 0.5657
\]
(Figure~\ref{fig:solplus} depicts~$\gamma_i(t)$ and~$e_i$ over three periods).
Thus,
\begin{equation}
\label{eq:pivr}
\GOE(u)=R_P-R_C = -0.0136.
\end{equation}

\begin{figure}[!t]
\centering
\includegraphics[scale=0.6]{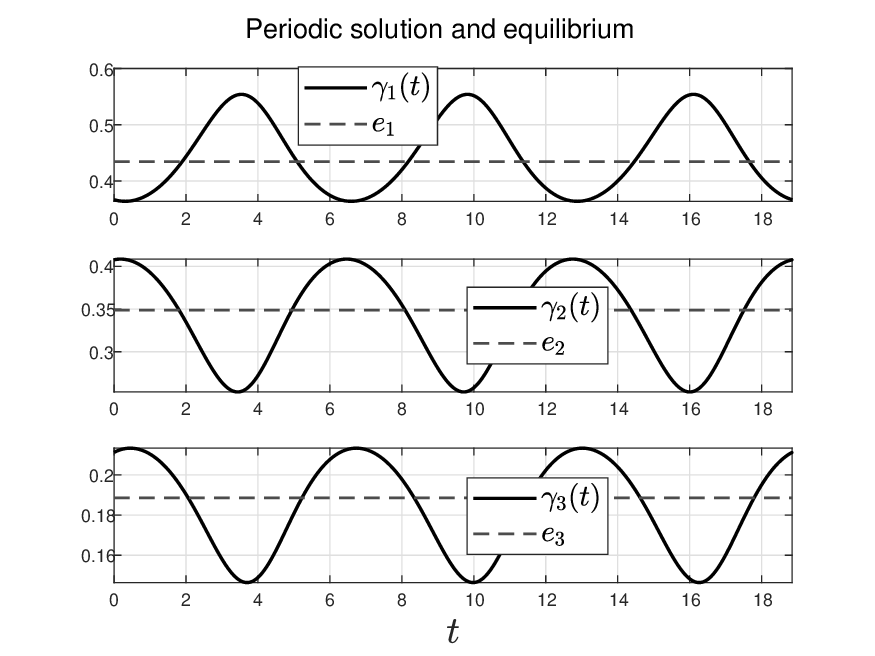}
\caption{  Periodic orbit~$\gamma$ and the equilibrium~$e$
as a function of~$t\in[0,3T)$. }
\label{fig:solplus}
\end{figure}

To demonstrate other quantities defined in the proof of the main result,
we also computed the~$p_i$s 
using~\eqref{eq:p-recursion} and~\eqref{eq:p-normalization}  yielding
\[
p_0= 0.4398,\qquad
p_1=0.3376,\qquad
p_2=0.1807,\qquad
p_3=0.0420  ,
\]
and numerically calculated  the~$d_i$s
(see Figure~\ref{fig:dis}). It may be seen that~$d_i(t)\geq 0$ for all~$i$ and~$t$. A numerical integration   
gives
\[
-\sum_{i=0}^{3}p_i\avg{u_i d_i}=-0.0135,
\]
and this agrees with~\eqref{eq:pivr} up to a small numerical error due to the  integrations.

\begin{figure}[!t]
\centering
\includegraphics[scale=0.6]{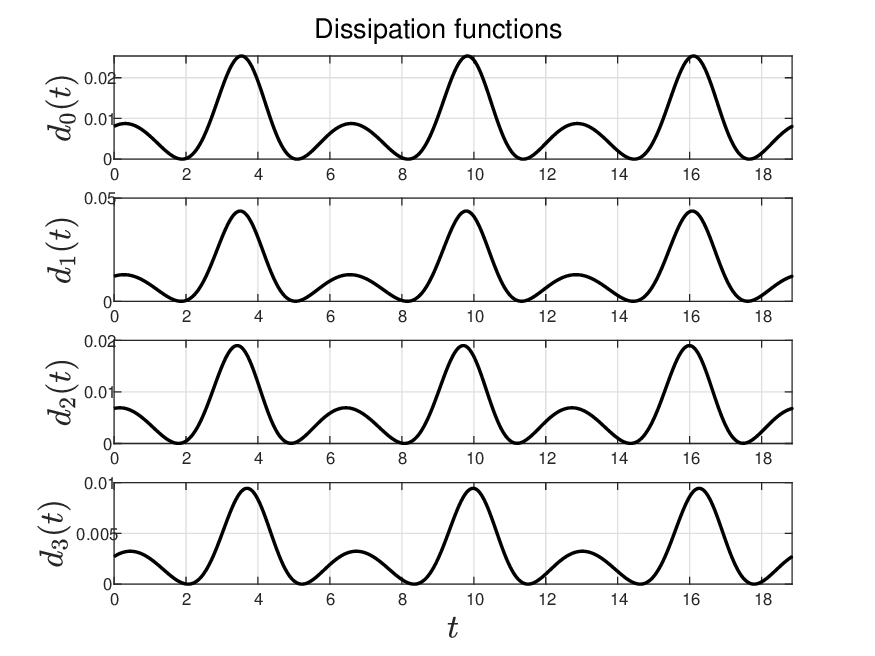}
\caption{ Periodic dissipation functions~$d_0(t),\dots,d_3(t) $
as a function of~$t\in[0,3T)$. }
\label{fig:dis}
\end{figure}

Note that we intentionally chose 
  a    simple numerical example, but Theorem~\ref{thm:main}  shows that~$\GOE(u)$ will be nonpositive  for any
dimension~$n$, any period~$T>0$, and any     positive and continuous $T$-periodic transition rates~$u$, 
 regardless of the  modulation amplitude,
 spatial location, etc.

\section{Discussion}
Many systems and processes 
can be regulated using both constant and periodic controls.  If the system entrains then the GOE provides a rigorous formulation for the following important  question: can 
entrainment also increase the output, on average.

Here, we analyzed~GOE in
the~RFM and proved that, perhaps surprisingly, it is always nonpositive. 
This holds  for every dimension~$n$  of the~RFM and every collection of positive, continuous, and 
jointly $T$-periodic rates. 

The proof is based on a new approach:   $\GOE(u)$
is represented through a weighted sum of edge-wise dissipation terms, with the weights constructed recursively from the equilibrium corresponding to the constant rates~$\avg{u}$.
We also derived a complete characterization of the case where   GOE is zero. This happens  if and only if 
all the transition rates share one common positive scalar modulation of their mean
values. Such a modulation is merely a time reparametrization of the autonomous
constant-rate dynamics. Thus, every genuinely nonuniform periodic modulation
  yields a   negative~GOE.
  Periodic  forcing can  entrain 
  the~RFM to  produce
  coordinated rhythmic behavior, but 
entrainment is not free. 
We refer to this fundamental limitation as the ``cost of entrainment''.
Our results suggest that periodic synchronization and performance enhancement are not automatically compatible.
This is potentially relevant to translation, circadian regulation, periodic resource availability, traffic flow, and more general
nonlinear transport systems. The RFM provides a mathematically tractable setting in which this tradeoff can be proved exactly.

The proof technique raises a broader  question: which nonlinear transport systems exhibit an analogous no-GOE property? Identifying the structural mechanisms underlying such behavior could help distinguish systems in which periodic forcing is intrinsically costly from those in which it can genuinely enhance throughput.



\end{document}